\documentclass{amsart}
\usepackage{txfonts}
\usepackage{amsfonts}
\usepackage{amssymb}
\usepackage{mathrsfs}
 \newtheorem{thm}{Theorem}[section]

 \newtheorem{prop}[thm]{Proposition}
 \newtheorem{lem}[thm]{Lemma}

  \newcommand{\f}{\mathbb{F}_{q}}
  \newcommand{\fl}{\mathbb{F}_{q^{l}}}
   
\begin{document}
\title{Security Analysis on `` An Authentication Code Against Pollution Attacks in Network Coding''}
\author{Jun Zhang, Xinran Li and Fang-Wei Fu}
\address{Chern Institute of Mathematics, Nankai University, Tianjin, P.R. China}
\email{zhangjun04@mail.nankai.edu.cn; xinranli@mail.nankai.edu.cn; fwfu@nankai.edu.cn}
\thanks{This research is supported by the National Key Basic Research Program of China (Grant No. 2013CB834204), and the National Natural Science Foundation of
China (Nos. 61171082, 10990011 and 60872025). The author Jun Zhang is also supproted by the Chinese Scholarship Council under the State Scholarship Fund during visiting University of California, Irvine.}
\date{}
 \maketitle
\begin{abstract}
We analyze the security of the authentication code against pollution attacks in network coding given by Oggier and Fathi~\cite{oggier} and show one way to remove one very strong condition they required. Actually, we find a way to attack their authentication scheme. In their scheme, they considered that if some malicious nodes in the network collude to make pollution in the network flow or make substitution attacks to other nodes, they thought these malicious nodes must solve a system of linear equations to recover the secret parameters. Then they concluded that their scheme is an unconditional secure scheme. Actually, note that the authentication tag in the scheme of Oggier and Fathi is nearly linear on the messages, so it is very easy for any malicious node to make pollution attack in the network flow, replacing the vector of any incoming edge by linear combination of his incoming vectors whose coefficients have sum $1$. And if the coalition of malicious nodes can carry out decoding of the network coding, they can easily make substitution attack to any other node even if they do not know any information of the private key of the node. Moreover, even if their scheme can work fruitfully, the condition in their scheme $H\leqslant M$ in a network can be removed, where $H$ is the sum of numbers of the incoming edges at adversaries. Under the condition $H\leqslant M$, $H$ may be large, so we need large parameter $M$ which increases the cost of computation a lot. On the other hand, the parameter $M$ can not be very large as it can not exceed the length of original messages.

\end{abstract}
\section{Introduction}
Network coding is a novel technique to achieve the maximum multicast throughput, which was introduced by Ahlswede {\it et al.} \cite{ACL}. It allows the intermediate node to generate output data by mixing its received data. In 2003, Li {\it et al.} \cite{LY} further showed that linear network coding is sufficient to achieve the optimal throughput in multicast networks. Subsequently, Ho {\it et al.} \cite{HKM} introduced the concept of random linear network coding, and proved that it achieves the maximum throughput of multicast network with high probability. Network coding is efficiently applicable to numerous forms of network communications, such as Internet TV, wireless networks, content distribution networks and P2P networks. Due to these advantages, network coding attracts many researchers and has developed very quickly.

However, networks using network coding impose security problems that traditional networks do not face. A particularly important problem is the pollution attack. If some nodes in the network are malicious and inject corrupted packets into the information flow, then the honest intermediate node mix invalid packet with other packets.
 According to the rule of network coding, the corrupted outgoing packets quickly pollute the whole network and cause all the messages to be decoded wrongly in the destination.

Recently several related works are proposed to address the pollution attack, such as homomorphic hashing, digital signature and message authentication code (MAC). Krohn {\it et al.} \cite{KFM} (see also \cite{GR}) used homomorphic hashing function to prevent pollution attacks. Yu {\it et al.} \cite{YWR} proposed a homomorphic signature scheme based on discrete logarithm and RSA, which however was showed insecurely by Yun {\it et al.} \cite{YCK}. Charles {\it et al.} \cite{CJL} gave a signature scheme based on Weil pairing over elliptic curves and provided authentication of the data in addition to detecting pollution attacks. Zhao {\it et al.} \cite{ZKM} designed a signature scheme that view all blocks of the file as vectors and make use of the fact that all valid vectors transmitted in the network should belong to the subspace spanned by the original set of vectors from the file. Boneh {\it et al.} \cite{BFK} proposed two signature schemes that can be used in conjunction with network coding to prevent malicious modification of messages, and they showed that their constructions had a lower signature length compared with related prior work. Boneh {\it et al.} \cite{BF} constructed a linearly homomorphic signature scheme that authenticates vectors with coordinates in the binary field $\mathbb{F}_{2}$. It is the first such scheme based on the hard problem of finding short vectors in integer lattices. Agrawal and Boneh \cite{AB} designed a homomorphic MAC system that allows checking the integrity of network coded data. These works provide computational security (i.e., the attacker's resources are limited) in network coding.

Besides digital signatures and MACs, authentication codes also satisfy the properties of authentication. However, authentication code provides unconditional security (i.e., the attacker has unlimited computational power). In the multi-receiver authentication model, a sender broadcasts an authenticated message such that all the receivers can independently verify the authenticity of the message with their own private keys. It requires a security that malicious groups of up to a given size of receivers can not successfully impersonate the transmitter, or substitute a transmitted message. Desmedt {\it et al.} \cite{DFY} gave an authentication scheme of single message for multi-receivers. Safavi-Naini and Wang \cite{SW} extended the DFY scheme \cite{DFY} to be an authentication scheme of multiple messages for multi-receivers. Note that their construction was not linear over the base field with respect to the message. Oggier and Fathi \cite{OF,oggier} made a little modification of the construction so that the construction can be used for network coding, which is actually not secure we will show in this paper. Tang \cite{Tang} used homomorphic authentication codes to sign a subspace which provide an unconditionally security. In fact, Tang in the same paper had noticed that linear authentication codes for linear network is not secure, so he modified the type of substitution attack.

Firstly, we recall the general model of network coding and the definition of subspace codes.
In the basic multicast model for linear network coding, a source node $s$ generates $n$ messages, each consisting of $m$ symbols in the base field $\f$. Let $\{x_{1},x_{2},\ldots,x_{n}\}\subseteq\f^{l\times 1}$ represent the set of messages. Based on the messages, the source node $s$ transmits a message over each outgoing channel. At a node in the network, the symbols on its outgoing channel are $\f$-linear combinations of incoming symbols. For a node $i$, define $Out(i)=\{e\in E : e $ is an outgoing channel of $i\}$, and $In(i) = \{e \in E : e $ is an incoming channel of $i\}$. If the channel $e$ of network carries packet $y(e)$, where $e\in Out(i)$, and $i$ is an internal nodes, then $y(e)$ satisfies $y(e)=\sum_{d\in in(i)}k_{de}y(d)$. The $|In(i)|\times |Out(i)|$ matrix $K_{i}=[k_{de}]_{d\in in(i), e\in Out(i)}$ is called the \emph{local encoding kernel at node} $i$. Note that each $y(e)$ is a linear combination of the messages sent by the source node, so there exists a vector $f_{e}\in \f^{1\times n}$ such that
\begin{equation*}
    y(e)= f_{e}\underline{\mathbf{X}},\,\mbox{where}\,\underline{\mathbf{X}}=\left(
                                                                               \begin{array}{c}
                                                                                 x_1 \\
                                                                                 x_{2}\\
                                                                                 \vdots\\
                                                                                 x_{n} \\
                                                                               \end{array}
                                                                             \right)\ .
\end{equation*}
The vector $f_{e}$ is called the \emph{global encoding vector} of channel $e$. Given the local encoding kernels for all the channels in network, the global encoding kernels can be calculated recursively in any upstream-to-downstream order as follows
$$f_e=\sum_{d\in in(i)}k_{de}f_d\ .$$
 Write the received vectors at a node $t$ as a column vector
 \begin{equation*}
    {A_{t}=(y(e)\ : e\in In(t))^{T}}=\left(
                                       \begin{array}{c}
                                         y(e_{1}) \\
                                         y(e_2) \\
                                         \vdots\\
                                         y(e_{e(t)}) \\
                                       \end{array}
                                     \right)\ ,
 \end{equation*}
 where $In(t)=\{e_{1},e_2,\cdots,e_{e(t)}\}$.
Then we have the decoding equation at the node $t$
 $$ { F_{t}\cdot\mathbf{\underline{X}}=A_{t}\ ,}$$
where
 \begin{equation*}
   F_{t}=(f_{e}: e\in In(t))^{T}=\left(
                                       \begin{array}{c}
                                         f_{e_{1}} \\
                                         f_{e_2} \\
                                         \vdots\\
                                         f_{e_{e(t)}} \\
                                       \end{array}
                                     \right)
 \end{equation*}
is called the \emph{global encoding kernel} at the node $t$.

\section{The Authentication Scheme of Oggier and Fathi}
Oggier and Fathi constructed an authentication code against pollution and substitution attacks in network coding, and they proved that the scheme is unconditional secure under some condition. Let us recall their construction and their result about the security analysis.
\begin{itemize}
  \item \textbf{Key generation:} A trusted authority randomly generates $M+1$ polynomials $P_{0}(x),P_1(x),$ $\cdots,P_{M}(x)\in \fl[x]$ and choose $V$ distinct values $x_{1},\cdots,x_V\in \fl$. These polynomials are of degree $k-1$, and we denote them by
      \[
       P_{i}(x)=a_{i,0}+a_{i,1}x+a_{i,2}x^2+\cdots+a_{i,k-1}x^{k-1},\quad i=0,1,\cdots,M\ .
      \]
  \item \textbf{Key distribution:} The trusted authority gives as private key to the source $S$ the $M+1$ polynomials $(P_{0}(x),\cdots,P_{M}(x))$, and as private key for each verifier $R_i$ the $M+1$ valuations of polynomials at $x=x_i$, namely $(P_{0}(x_i),\cdots,P_{M}(x_i))$, $i=1,2,\cdots,V$. The values $x_{1},\cdots,x_V$ are made public. The keys can be given to the nodes when they sign up for a service protected by this scheme.
  \item \textbf{Authentication tag:} Let us assume that the source wants to send $n$ data messages $s_1,s_2,\cdots,s_n\in\f^l$. Choose and fix an $\f$-linear isomorphism between $\f^l$ and $\fl$, then consider they have the same elements. The source computes the following polynomial in $\fl[x]$:
      \[
        A_{s_i}(x)=P_{0}(x)+s_iP_{1}(x)+s_i^{q}P_{2}(x)+\cdots+s_i^{q^{M-1}}P_{M}(x)
      \]
      which forms the authentication tag of each $s_i$, $i=1,\cdots,n$. Instead of sending the original messages $s_1,s_2,\cdots,s_n$, the source actually sends packets $\vec{x}_i$ of the form
      \[
        \vec{x}_i=[1,s_i,A_{s_i}(x)]\in\f^{1+l+kl},\qquad i=1,\cdots,n\ .
      \]
\end{itemize}

The security of the authentication scheme above proven by Oggier and Fathi is as follows:
\begin{prop}[\cite{oggier}]
Consider a multicast network implementing linear network coding, among which nodes $V$ of them are verifying nodes owning a private key for authentication.
The above scheme is an unconditionally secure network coding authentication code against a coalition of up to $k-1$ adversaries, possibly among the verifying nodes, in which every key can be used to authentication up to $M$ messages, under the assumption that $H\leqslant M$, where $H$ is the sum of numbers of the incoming edges at each adversary.
\end{prop}

\section{Linear Substitution/Pollution Attacks to their Scheme}
In the security analysis given by Oggier and Fathi, they focused on solving the system of linear equations on variables $a_{i,j}$ to recover the private key of other node. Actually, notice that the authenticated vectors $\vec{x}_i$ above are nearly linear on messages, so we can implement linear substitution attack to their scheme. In some papers[??], they have noticed that it is not secure to use linear authentication codes on linear network. And they considered a new type of substitution attack. Also, they pointed out that the authentication code of Oggier and Fathi is non-linear so that it should be still secure. Next, we present our linear substitution attack in details.

 Suppose the coalition of malicious verifying nodes can carry out decoding of the network coding, i.e., the coalition of their global kernels                                                         has rank not less than the                                                                                                                                                                      minimum cut of the network, for instance, the coalition of malicious verifying nodes contains one destination node. In this case, they can                                                         decode the tagged messages sent by the source node:
\[
   \vec{x}_i=[1,s_i, A_{s_i}(x)] \qquad \mbox{for $i=1,2,\cdots,n$.}
 \]
 For any $a_1,a_2,\cdots,a_{n}\in \f$ such that
 \[
      a_1+a_2+\cdots+a_{n}=1\ ,
 \]
 replace $\vec{x}_n$ by $\vec{x}'_n=\sum_{i=1}^{n}a_i\vec{x}_i$. Next, we show that in this way each verifying node can not notice this
 substitution attack.

 \begin{flushleft}
  \emph{Verification of Linear Substitution Attack:}
\end{flushleft}
The vector of any incoming edge at any node is of the form
 \[
       \sum_{i=1}^{n-1}\alpha_{i}\vec{x}_i+\alpha_{n}\vec{x}'_n=[\sum_{i=1}^{n}\alpha_{i},\sum_{i=1}^{n-1} \alpha_is_i+\alpha_{n}s'_n,\sum_{i=1}^{n-1} \alpha_iA_{s_i}(x)+\alpha_{n}A_{s'_n}(x)]
 \]
 for some $\alpha_1,\alpha_2,\cdots,\alpha_n\in \f$. Then
   \begin{equation*}
       \begin{array}{rl}
          & \sum_{i=1}^{M} \left(\sum_{j=1}^{n-1} \alpha_js_j+\alpha_{n}s'_n\right)^{q^{i-1}} P_{i}(x)+P_{0}(x)(\sum_{j=1}^{n} \alpha_j) \\
         = & \sum_{i=1}^{M} \left(\sum_{j=1}^{n-1} \alpha_js_j^{q^{i-1}}+\alpha_{n}s_n^{\prime q^{i-1}}\right)
         P_{i}(x)+P_{0}(x)(\sum_{j=1}^{n} \alpha_j) \\
         = &  \sum_{i=1}^{M} \left(\sum_{j=1}^{n-1} \alpha_js_j^{q^{i-1}}+\alpha_{n}\sum_{t=1}^{n} a_ts_t^{q^{i-1}}\right)
         P_{i}(x)+P_{0}(x)(\sum_{j=1}^{n} \alpha_j) \\
         = & \sum_{i=1}^{M} P_{i}(x) \left(\sum_{j=1}^{n-1} \alpha_js_j^{q^{i-1}}\right)+\alpha_{n} \sum_{i=1}^{M}P_{i}(x)\left(\sum_{t=1}^{n}
         a_ts_t^{q^{i-1}}\right)\\
         &+P_{0}(x)(\sum_{j=1}^{n} \alpha_j)
       \end{array}
   \end{equation*}
 equals to
  \begin{equation*}
     \begin{array}{rl}
         & \sum_{i=1}^{n-1} \alpha_iA_{s_i}(x)+\alpha_{n}A_{s'_n}(x) \\
       = & \sum_{i=1}^{n-1} \alpha_i\left(P_0(x)+\sum_{t=1}^Ms_i^{q^{t-1}}P_t(x)\right)\\
       &+\alpha_{n}\left((\sum_{j}^{n-1}a_j)P_0(x)+\sum_{t=1}^Ms_n^{\prime q^{t-1}}P_t(x)\right) \\
       = & (\sum_{i=1}^{n} \alpha_i)P_0(x)+\sum_{i=1}^{n-1}\sum_{t=1}^M\alpha_is_i^{q^{t-1}}P_t(x)\\
       &+\alpha_{n}\sum_{t=1}^M\left(\sum_{j=1}^{n}a_js_j^{q^{t-1}}\right)P_t(x) \\
       = & (\sum_{i=1}^{n} \alpha_i)P_0(x)+ \sum_{t=1}^M \left(\sum_{i=1}^{n-1} \alpha_is_i^{q^{t-1}}\right)P_t(x)\\
       &+\alpha_{n}\sum_{t=1}^M\left(\sum_{j=1}^{n}a_js_j^{q^{t-1}}\right)P_t(x)
     \end{array}
  \end{equation*}
   for all $x\in \fl$. In other words, it can be verified by any verifying node using his private key.

 From the above argument, we can see that any node in the network can easily make pollution to the network flow in the way that the node replaces any one or more of the vectors he received by linear combinations of his incoming vectors whose coefficients have sum $1$ and then the node processes the network coding with the new vectors.

Finally, we point out that even if Oggier and Fathi's scheme can work fruitfully, the condition $H\leqslant M$ can also be removed. Note that the condition $H\leqslant M$ is very critical in a network. The proof is similar to the proof given by Oggier and Fathi. They wrote the secret parameters $A=(a_{i,j})$ as a column vector in the order as following
\[
   \vec{a}=(a_{0,1},a_{0,2},\cdots,a_{0,k},a_{1,1},\cdots,a_{1,k},\cdots,a_{M,1},a_{M,2},\cdots,a_{M,k})^T\ ,
\]
where $G^T$ represents the transpose of the matrix $G$, and they rewrote the system of linear equations using $\vec{a}$. Then they computed the rank of the coefficient matrix, finally they concluded that under the condition $H\leqslant M$ the rank of the coefficient matrix is less than the number of variables $k(M+1)$. Actually, if we rewrite the secret parameters $A=(a_{i,j})$ as a column vector in the following order
\[
   \vec{a}'=(a_{0,1},a_{1,1},\cdots,a_{M,1},a_{0,2},\cdots,a_{M,2},\cdots,a_{0,k},a_{1,k},\cdots,a_{M,k})^T\ .
\]
Then we obtain a new system of linear equations on $a_{i,j}$ using $\vec{a}'$. In this way, we can easily show that the rank of the coefficient matrix is always less than the number of variables. So the system of linear equations does always have solutions. Next, we give the details.

Suppose a group of $K$ malicious nodes collaborate to recover $A$ and make a substitution attack. Without loss of generality, we assume that the malicious nodes are $R_1,R_2,\cdots,R_K$. Suppose the global encoding kernel at the verifying node $R_i$ is
  \begin{equation*}
    H_i=\left(
  \begin{array}{cccc}
    h_{1,1}^{(i)} & h_{1,2}^{(i)} & \cdots & h_{1,n}^{(i)} \\
     h_{2,1}^{(i)} & h_{2,2}^{(i)} & \cdots & h_{2,n}^{(i)} \\
     \vdots & \vdots & \ddots & \vdots \\
    h_{e(i),1}^{(i)} & h_{e(i),2}^{(i)} & \cdots & h_{e(i),n}^{(i)}\\
  \end{array}
\right)\ .
\end{equation*}
Each $R_i$ has some information about the secret parameter matrix $A=(a_{i,j})$:
\begin{gather*}\label{equation1}
   \left(
     \begin{array}{ccccc}
      \sum_{j=1}^n h_{1,j}^{(i)}& \sum_{j=1}^n h_{1,j}^{(i)}s_j & \sum_{j=1}^n h_{1,j}^{(i)}s_j^q & \cdots & \sum_{j=1}^n h_{1,j}^{(i)}s_j^{q^{M-1}} \\
      \sum_{j=1}^n h_{2,j}^{(i)}&  \sum_{j=1}^n h_{2,j}^{(i)}s_j & \sum_{j=1}^n h_{2,j}^{(i)}s_j^q & \cdots & \sum_{j=1}^n h_{2,j}^{(i)}s_j^{q^{M-1}}\\
      \vdots & \vdots & \vdots & \ddots & \vdots \\
      \sum_{j=1}^n h_{e(i),j}^{(i)}& \sum_{j=1}^n h_{e(i),j}^{(i)}s_j & \sum_{j=1}^n h_{e(i),j}^{(i)}s_j^q & \cdots & \sum_{j=1}^n h_{e(i),j}^{(i)}s_j^{q^{M-1}}\\
     \end{array}
   \right)\cdot A\\=
   \left(
     \begin{array}{cccc}
       \sum_{j=1}^n h_{1,j}^{(i)}L_1(s_j) & \sum_{j=1}^n h_{1,j}^{(i)}L_2(s_j) & \cdots & \sum_{j=1}^n h_{1,j}^{(i)}L_k(s_j) \\
        \sum_{j=1}^n h_{2,j}^{(i)}L_1(s_j) & \sum_{j=1}^n h_{2,j}^{(i)}L_2(s_j) & \cdots & \sum_{j=1}^n h_{2,j}^{(i)}L_k(s_j)\\
       \vdots & \vdots & \ddots & \vdots \\
       \sum_{j=1}^n h_{e(i),j}^{(i)}L_1(s_j) & \sum_{j=1}^n h_{e(i),j}^{(i)}L_2(s_j) & \cdots & \sum_{j=1}^n h_{e(i),j}^{(i)}L_k(s_j)\\
     \end{array}
   \right)
\end{gather*}
and
\begin{gather*}
    A\cdot \left(
                       \begin{array}{c}
                          1 \\
                          x_i \\
                          \vdots\\
                          x_i^{k-1} \\
                        \end{array}
                      \right)=\left(
                       \begin{array}{c}
                         P_0(x_i) \\
                         P_1(x_i) \\
                          \vdots\\
                         P_M(x_i) \\
                        \end{array}
                      \right)\ .
\end{gather*}
The group of malicious nodes combines their equations, and they get a system of linear equations
\begin{equation}\label{equation}
\left\{
  \begin{array}{c}
    \left(
                     \begin{array}{c}
                       D_1 \\
                       \vdots \\
                       D_K \\
                     \end{array}
                   \right)\cdot A=\left(
                                    \begin{array}{c}
                                      C_1 \\
                                      \vdots \\
                                      C_K \\
                                    \end{array}
                                  \right),
\\
    A\cdot \left(
  \begin{array}{cccc}
    1 & 1 & \cdots & 1 \\
     x_1^q & x_2^q & \cdots & x_K^q \\
     \vdots & \vdots & \ddots & \vdots \\
    x_1^{q^{k-1}} & x_2^{q^{k-1}} & \cdots & x_K^{q^{k-1}}\\
  \end{array}
\right)=\left(
  \begin{array}{cccc}
    P_0(x_1) & P_0(x_2) & \cdots & P_0(x_K) \\
     P_1(x_1) &  P_1(x_2) & \cdots &  P_1(x_K) \\
     \vdots & \vdots & \ddots & \vdots \\
     P_M(x_1) &  P_M(x_2) & \cdots &  P_M(x_K)\\
  \end{array}
\right)\ ,
  \end{array}
\right.
\end{equation}
where
\begin{equation*}
    D_i=\left(
     \begin{array}{ccccc}
      \sum_{j=1}^n h_{1,j}^{(i)}& \sum_{j=1}^n h_{1,j}^{(i)}s_j & \sum_{j=1}^n h_{1,j}^{(i)}s_j^q & \cdots & \sum_{j=1}^n h_{1,j}^{(i)}s_j^{q^{M-1}} \\
       \sum_{j=1}^n h_{2,j}^{(i)}&  \sum_{j=1}^n h_{2,j}^{(i)}s_j & \sum_{j=1}^n h_{2,j}^{(i)}s_j^q & \cdots & \sum_{j=1}^n h_{2,j}^{(i)}\vec{s}_j^{q^{M-1}}\\
      \vdots & \vdots & \vdots & \ddots & \vdots \\
      \sum_{j=1}^n h_{e(i),j}^{(i)}&  \sum_{j=1}^n h_{e(i),j}^{(i)}s_j & \sum_{j=1}^n h_{e(i),j}^{(i)}s_j^q & \cdots & \sum_{j=1}^n h_{e(i),j}^{(i)}s_j^{q^{M-1}}\\
     \end{array}
   \right)
\end{equation*}
and
\begin{equation*}
C_i= \left(
     \begin{array}{cccc}
       \sum_{j=1}^n h_{1,j}^{(i)}L_1(s_j) & \sum_{j=1}^n h_{1,j}^{(i)}L_2(s_j) & \cdots & \sum_{j=1}^n h_{1,j}^{(i)}L_k(s_j) \\
        \sum_{j=1}^n h_{2,j}^{(i)}L_1(s_j) & \sum_{j=1}^n h_{2,j}^{(i)}L_2(s_j) & \cdots & \sum_{j=1}^n h_{2,j}^{(i)}L_k(s_j)\\
       \vdots & \vdots & \ddots & \vdots \\
       \sum_{j=1}^n h_{e(i),j}^{(i)}L_1(s_j) & \sum_{j=1}^n h_{e(i),j}^{(i)}L_2(s_j) & \cdots & \sum_{j=1}^n h_{e(i),j}^{(i)}L_k(s_j)\\
     \end{array}
   \right)\ .
\end{equation*}
Denote
\begin{equation*}
 S_n=   \left(
      \begin{array}{ccccc}
        1 & s_1 & s_1^{q} & \cdots & s_1^{q^{M-1}} \\
        1 & s_2 & s_2^{q} & \cdots & s_2^{q^{M-1}} \\
        \vdots & \vdots & \vdots & \ddots & \vdots \\
         1 & s_n & s_n^{q} & \cdots & s_n^{q^{M-1}} \\
      \end{array}
    \right) \ .
\end{equation*} Then
\[
  D_i=H_i\cdot S_n\ .
\]

\begin{lem}\label{keylem}
If $K\leqslant k-1$, then there exists exact $q^{l(M+1-r_0)(k-K)}$ matrices $A$ satisfying the system of equations (\ref{equation}), where
  \begin{equation*}
  r_0=\mathrm{rank}\left(\left(
      \begin{array}{c}
        H_{1}S_n \\
        H_{2}S_n \\
        \vdots \\
        H_{K}S_n \\
      \end{array}
    \right)\right)\ .
 \end{equation*}
\end{lem}
\begin{proof}
Recall the system~(\ref{equation})
\begin{equation*}
\left\{
  \begin{array}{l}
    \left(
                     \begin{array}{c}
                       H_1S_n \\
                       \vdots \\
                       H_KS_n \\
                     \end{array}
                   \right)\cdot A=\left(
                                    \begin{array}{c}
                                      C_1 \\
                                      \vdots \\
                                      C_K \\
                                    \end{array}
                                  \right),
\\
    A\cdot  \left(
  \begin{array}{cccc}
    1 & 1 & \cdots & 1 \\
     x_1^q & x_2^q & \cdots & x_K^q \\
     \vdots & \vdots & \ddots & \vdots \\
    x_1^{q^{k-1}} & x_2^{q^{k-1}} & \cdots & x_K^{q^{k-1}}\\
  \end{array}
\right)=\left(
  \begin{array}{ccc}
    P_0(x_1)  & \cdots & P_0(x_K) \\
     P_1(x_1) & \cdots &  P_1(x_K) \\
     \vdots &  \ddots & \vdots \\
     P_M(x_1)  & \cdots &  P_M(x_K)\\
  \end{array}
\right)\ .
  \end{array}
\right.
\end{equation*}
Rewrite the matrix $A$ of variables as a single column of $k(M+1)$ variables. Then the system~(\ref{equation}) becomes
 \begin{equation}\label{equ2}
    \left(
      \begin{array}{cccc}
        H_1S_n & 0 & 0 & 0 \\
        0 & H_1S_n & 0 & 0 \\
        0 & 0 & \ddots & 0 \\
        0 & 0 & 0 & H_1S_n \\
        \vdots & \vdots & \ddots & \vdots \\
         H_KS_n & 0 & 0 & 0 \\
         0 & H_KS_n & 0 &0  \\
         0 & 0 & \ddots & 0 \\
         0 & 0 & 0 & H_KS_n \\
        I_{M+1} & x_1 I_{M+1} & \cdots & x_1^{k-1}I_{M+1} \\
        I_{M+1} & x_2I_{M+1} & \cdots & x_2^{k-1}I_{M+1} \\
        \vdots & \vdots & \ddots & \vdots \\
        I_{M+1} & x_K I_{M+1} & \cdots & x_K^{k-1}I_{M+1} \\
      \end{array}
    \right)\cdot
    \left(
      \begin{array}{c}
        a_{0,1} \\
        a_{1,1} \\
        \vdots \\
        a_{M,1} \\
        a_{0,2}\\
        a_{1,2}\\
        \vdots \\
        a_{M,2}\\
        \vdots \\
         a_{0,k}\\
         a_{1,k}\\
         \vdots \\
         a_{M,k}\\
      \end{array}
    \right)=T
    \end{equation}
where $I_{M+1}$ is the identity matrix with rank ($M+1$) and $T$ is the column vector of the constant terms in system (\ref{equation}) with proper order. Notice that
\begin{equation*}
   r_0= \mathrm{rank}\left(\left(
                         \begin{array}{c}
                           H_1S_n \\
                           H_2S_n \\
                           \vdots \\
                           H_KS_n \\
                         \end{array}
                       \right)\right) =  \mathrm{rank}\left(\left(
                                                       \begin{array}{c}
                                                        H_1 \\
                                                        H_2 \\
                                                        \vdots \\
                                                        H_K \\
                                                       \end{array}
                                                     \right) \cdot S_n \right)
                                                     \leqslant \min\left\{\mathrm{rank}\left(
                                                                                               \begin{array}{c}
                                                                                                H_1 \\
                                                                                                H_2 \\
                                                                                                \vdots \\
                                                                                                H_K \\
                                                                                               \end{array}
                                                                                               \right), n  \right\} \ .
      \end{equation*}
 Also note that rows of
 \begin{equation*}
  \left(
      \begin{array}{c}
        H_1S_n \\
        H_2S_n \\
        \vdots \\
        H_KS_n \\
      \end{array}
    \right)
 \end{equation*}
  is contained in the space $\f^{M+1}$ generated by $x_i^jI_{M+1}$ if $x_i\neq 0$.
  So the rank of the coefficient matrix of System~(\ref{equ2}) to
  \[
      r_0k+(M+1-r_0)K
  \]
which is less than the number of variables $k(M+1)$. So the system (\ref{equ2}) has
$$q^{l(k(M+1)-( r_0k+(M+1-r_0)K))}=q^{l(M+1-r_0)(k-K)}$$
solutions, i.e., the system (\ref{equation}) has $q^{l(M+1-r_0)(k-K)}$ solutions.
\end{proof}

\section{Conclusion}
In this paper, we discuss the security of the authentication code given by Oggier and Fathi and show our linear attack to their scheme, although it looks like non-linear. So we point out that as the technique of linear network develops very fast, and it has invaded a lot in our daily life, such as Internet TV, wireless networks, content distribution networks, P2P networks and distributed file system, to give an efficient and unconditional secure authentication code for linear network against the original substitution/pollution attack considered by Oggier and Fathi is extremely urgent.


\bibliographystyle{ieeetran}
\bibliography{authentication}

 \end{document}